\documentclass[11pt]{article}
\usepackage[margin=1in]{geometry}
\usepackage{amsmath}
\usepackage{amssymb}
\usepackage{amsthm}
\usepackage{bm}
\usepackage{cite}
\usepackage{graphicx}
\usepackage{subfig}
\usepackage[usenames,dvipsnames]{xcolor}
\usepackage{hyperref}
\usepackage[capitalise]{cleveref}

\newtheorem{proposition}{Proposition}
\newtheorem{corollary}[proposition]{Corollary}

\title{Power Law Public Goods Game for Personal Information Sharing in News Commentaries\footnote{Authors are listed in alphabetical order.}}

\author{Christopher Griffin\footnote{C. Griffin is with the Applied Research Laboratory, Penn State University, University Park PA, 16802. E-mail: \href{mailto:griffinch@ieee.org}{griffinch@ieee.org}} \and Sarah Rajtmajer\footnote{S. Rajtmajer, A. Squicciarini and P. Umar are with the College of Information Science and Technology, Penn State University, University Park PA, 16802. E-mail:,\href{mailto:smr48@psu.edu}{smr48@psu.edu}, \href{mailto:asquicciarini@ist.psu.edu}{asquicciarini@ist.psu.edu},\href{mailto:pxu3@ist.psu.edu}{pxu3@ist.psu.edu}} \and Anna Squicciarini\footnotemark[3] \and Prasanna Umar\footnotemark[3]}

\date{\today}

\begin{document}
\maketitle              
\begin{abstract}
We propose a public goods game model of user sharing in an online commenting forum. In particular, we assume that users who share personal information incur an information cost but reap the benefits of a more extensive social interaction. Freeloaders benefit from the same social interaction but do not share personal information. The resulting public goods structure is analyzed both theoretically and empirically. In particular, we show that the proposed game always possesses equilibria and we give sufficient conditions for pure strategy equilibria to emerge. These correspond to users who always behave the same way, either sharing or hiding personal information. We present an empirical analysis of a relevant data set, showing that our model parameters can be fit and that the proposed model has better explanatory power than a corresponding null (linear) model of behavior. 
\end{abstract}

\section{Introduction}
Recent work acknowledges the importance of online social engagement, noting that the bidirectional communication of the Internet allows readers to engage directly with reporters, peers, and news outlets to discuss issues of the day \cite{N09,S11,KPL16}. In parallel, studies have noted several challenges linked with this new form of readership, particularly the high level of toxicity and pollution from trolls and even bots often observed in these commentaries \cite{Bishop14}. While the negative impacts of trolling and abuse are well-studied (see, e.g., \cite{Forbes18,BTP14}), little attention has been paid to other more subtle risks involved with online commenting, particularly with respect to users' privacy. In particular, as users engage in discussion online, they often resort to \emph{self-disclosure} as a way to enhance immediate social rewards \cite{HZ17}, increase legitimacy and likeability \cite{BKO12}, or derive social support \cite{TW02}. By self-disclosure, we refer to the (possibly unintentional) act of disclosing identifying (e.g., location, age, gender, race) or sensitive (e.g., political affiliation, religious beliefs, cognitive and/or emotional vulnerabilities) personal information \cite{USR19}. 

In this work, we model the behavior of users posting comments about newspaper articles on major news platforms (e.g., NYT, CNN). We hypothesize that all users who participate in commentary about an article receive a ``reward'' that is proportional to the number of total comments posted; i.e., the net amount of social engagement generated. Hence, the act of self-disclosing comes at an information cost to the individual user yet may serve to increase the net return (e.g., total number of comments or impact the conversation in some capacity) all users receive. Accordingly, this scenario can be envisaged as a \emph{public goods game} in which pay-in is measured in terms of personal information and pay-out is measured in net quantity of social interaction through a commenting system. 



Public goods games are mathematical representations of the Tragedy of the Commons \cite{Lloyd1833,Hardin68} in which individuals must contribute to a common good in order to prevent that good from collapsing. Within a public goods game, cheating or freeloading is generally a more profitable choice; in this way, it is intellectually similar to the prisoner's dilemma (see, e.g., \cite{FT91,Bra04}), and various approaches to resolving the tragedy have been taken (e.g., \cite{Levin2014}). Public goods games have been widely studied as models of cooperation. In \cite{Led95}, the public goods game poses the following dilemma to a group of $N$ agents: each agent is asked to contribute $c$ monetary units towards a public good. Contributions earn a linear rate of return $r$, providing $rc$ monetary units for sharing. Thus, if $k$ individuals contribute, a contributing individual receives $rck/N - c$ monetary units, while a non-contributing individual receives $rck/N$ monetary units. Rational agents choose not to contribute.

There are several extensions to the classical public goods framework discussed above. Archetti and Scheuring \cite{AS12} and Young and Belmonte \cite{YB19} use a non-linear (power law) form of the public goods return function. We adopt this model in \cref{sec:Model}. Cooperation in a public goods setting is difficult to explain using a rational agent assumption and several approaches have been used to explain it. Volunteering in public goods is considered in \cite{HDHS02}. Punishment as a form of cooperation enforcement is discussed in \cite{HTBH08,FG00}. Reputation in an evolutionary public goods game is considered in \cite{H10}. The approach we take in this paper is substantially simpler; as we discuss in \cref{sec:Model}, we assume that each agent has a distinct information sharing cost, which leads to the emergence of equilibria in which users will share. Primary contributions of this paper are: development of a public goods model of personal information disclosure; proof of sufficient conditions for this game to to exhibit pure strategy equilibria; proof of existance of at least one equilibrium for any choice of model parameters; identification of necessary and sufficient conditions in specific cases; and, initial validation of the proposed model in a dataset of online comments on news articles.

The remainder of this paper is organized as follows: In \cref{sec:Data}, we discuss the data set used for model development and testing. We present our proposed model in \cref{sec:Model}. Mathematical analysis of the model is performed in \cref{sec:Math}. Experimental evidence supporting our model of user behavior is presented in \cref{sec:Experiment}. Finally, we provide conclusions and future directions in \cref{sec:Conclusion}.

\section{Data Description}\label{sec:Data}
We consider a set of user comments on news articles from four major English news websites \cite{barua2015tide}. The data set is composed of $59,249$ comments made by $22,132$ distinct users from March through August 2015. Comments are distributed across 2202 articles from The Huffington Post (1136), Techcrunch (119), CNBC (421) and ABC News (526). On average, each user contributes $2.68$ comments and participates in discussions related to $1.77$ articles. 

We use the unsupervised detection of self-disclosure proposed and validated in our earlier work \cite{USR19} to label these comments. Each comment is labeled for the presence or absence of self-disclosure, and each incidence of self-disclosure is tagged by category. We determine $10,858$ of the total $59,249$ comments to be self-disclosing.  Methods and initial results for self-disclosure tagging on this data set, including a breakdown of self-disclosures by category, are discussed in \cite{USR19}.

\section{Model}\label{sec:Model}
Let $R_i$ be the total number of comments associated with article $i$. This is also the common reward to \textit{all} commenters regardless of whether they provide personal information. Define the binary variable $\delta_k = 1$ if and only if User $k$ provides personal information in a comment \textit{at least once}. Using a public goods framework, we hypothesize the relationship:
\begin{equation}
R_i \sim A \cdot \left(\sum_k \delta_{i_k}\right)^\gamma + \epsilon_i,
\label{eqn:CommonReward}
\end{equation}
where $\gamma$ is a scaling factor and $A$ is constant of proportionality. The quantity $\epsilon_i$ is the (normally distributed) error associated with Article $i$. The individual payoff  to users in this pubic goods framework is:
\begin{equation}
r_{i_k} = R_i - \beta_k\delta_k,
\label{eqn:IndividualReward}
\end{equation}
where $\beta_k$ measures the sensitivity to information sharing for User $k$. In a totally symmetric game, $\beta = \beta_1 = \cdots = \beta_N$. 

\section{Mathematical Analysis}\label{sec:Math}
We analyze the model assuming that:
\begin{equation}
\delta_j \sim \mathrm{Bernoulli }(x_j),
\end{equation}
where $x_j \in [0,1]$ is the probability that user $j$ will disclose personal information. In a simultaneous game with $n$ users, each user will selfishly maximize her expected reward, which can be computed on the interior of the feasible region as:
\begin{equation}
u_j = \mathbb{E}(r_j) = \sum_{\delta_1\in\{0,1\}} \cdots \sum_{\delta_n \in \{0,1\}} A\left(\sum_k \delta_k\right)^\gamma\prod_k x_k^{\delta_k}(1-x_k)^{1-\delta_k} - \beta_jx_j.
\end{equation}
If for any $k$, $x_k$ is a pure strategy, then $\delta_k = x_k$ and $u_j$ is modified in the obvious way to prevent expressions of the form $0^0$. In particular, if $\mathcal{B} = \{0,1\}$ and $\mathbf{x} \in \mathcal{B}^n$ is pure, then:
\begin{equation}
u_j = A\left(\sum_k x_k\right)^\gamma - \beta_jx_j =  A\left(\sum_k \delta_k\right)^\gamma - \beta_j\delta_j. 
\label{eqn:PureStrategies}
\end{equation}
Put more simply, this is just an $n$-player, $n$-array (tensor) game, where each player has two strategies: disclose or don't disclose. The payoff structure is given by $n$ multi-linear maps:
\begin{displaymath}
\mathbf{A}^{(j)}:\underbrace{\mathbb{R}^2 \times \cdots \times \mathbb{R}^2}_n \to \mathbb{R}.
\end{displaymath}
The following result is guaranteed by Wilson's extension \cite{Wil71} of Nash's theorem \cite{Nash50} and the Lemke-Howson theorem \cite{LH61}:
\begin{proposition} There is at least one Nash equilibrium solution in simultaneous play. If the game is non-degenerate there are an odd number of equilibria.\hfill\qed
\end{proposition}
Fix the strategies for all users other than $j$ and denote this $\mathbf{x}_{-j}$. 
The (tensor) contraction $\mathbf{A}^{(j)}(\mathbf{x}_{-j})$ is a one-form (row vector). Assume:
\begin{equation}
\mathbf{A}^{(j)}(\mathbf{x}_{-j}) = \begin{bmatrix} C^{(j)}_1(\mathbf{x}_{-j}) &\;& C^{(j)}_0(\mathbf{x}_{-j}) \end{bmatrix}, 
\end{equation}
with:
\begin{align*}
C^{(j)}_1(\mathbf{x}_{-j}) &= \sum_{\bm{\delta}_{-j} \in \mathbb{B}^{n-1}}A\left(1 + \sum_{k\neq j} \delta_k\right)^\gamma\prod_{k \neq j}x_k^{\delta_k}(1-x_k)^{1-\delta_k} - \beta_j,\\
C^{(j)}_0(\mathbf{x}_{-j}) &= \sum_{\bm{\delta}_{-j} \in \mathbb{B}^{n-1}}A\left(\sum_{k\neq j} \delta_k\right)^\gamma\prod_{k \neq j}x_k^{\delta_k}(1-x_k)^{1-\delta_k}.
\end{align*}
As in \cref{eqn:PureStrategies}, care must be taken with this expression if $x_k$ is pure. 

If $\mathbf{x}_j = \langle{x_j, 1-x_j}\rangle$, then:
\begin{multline}
u_j(x_j, \mathbf{x}_{-j}) = \left\langle{\mathbf{A}^{(j)}(\mathbf{x}_{-j}),\mathbf{x}_j}\right\rangle = 
C^{(j)}_1(\mathbf{x}_{-j}) x_j + C^{(j)}_0(\mathbf{x}_{-j}) (1-x_j) =\\ \left(C^{(j)}_1(\mathbf{x}_{-j}) - C^{(j)}_0(\mathbf{x}_{-j})\right)x_j + C^{(j)}_0(\mathbf{x}_{-j}).
\label{eqn:ujC}
\end{multline}
A strategy vector $\mathbf{x} = (x_j, \mathbf{x}_{-j})$ is an equilibrium precisely when it solves the simultaneous optimization problem:
\begin{equation}
\forall j\left\{
\begin{aligned}
\max\;\; & u_j(x_j, \mathbf{x}_{-j})\\
s.t.\;\; & 0 \leq x_j \leq 1.
\end{aligned}
\right.
\end{equation}
We note the optimization problem for each Player $j$ is a linear programming problem.

\begin{proposition} A point $\mathbf{x}$ is an equilibrium if and only if there are vectors $\bm{\lambda},\bm{\mu} \in \mathbb{R}^n$ so that the following conditions hold:
\begin{equation}
\mathrm{PF}\left\{
\begin{aligned}
x_j - 1 &\leq 0 \qquad \forall j\\
-x_j &\leq 0 \qquad \forall j
\end{aligned}
\right.
\label{eqn:PF}
\end{equation}
\begin{equation}
\mathrm{DF}\left\{
\begin{aligned}
\frac{\partial u_j}{\partial x_j} + \lambda_j - \mu_j &= 0 \qquad \forall j\\
\lambda_j &\geq 0 \qquad \forall j\\
\mu_j &\geq 0 \qquad \forall j
\end{aligned}
\right.
\label{eqn:DF}
\end{equation}
\begin{equation}
\mathrm{CS}\left\{
\begin{aligned}
-\lambda_j x_j & = 0 \qquad \forall j\\
\mu_j (x_j - 1) & = 0 \qquad \forall j
\end{aligned}
\right..
\label{eqn:CS}
\end{equation}
Here:
\begin{equation}
\frac{\partial u_j}{\partial x_j} = C^{(j)}_1(\mathbf{x}_{-j}) - C^{(j)}_0(\mathbf{x}_{-j}).
\label{eqn:Partialuj}
\end{equation}
\label{prop:NeccessarySufficient}
\end{proposition}
\begin{proof} \cref{eqn:Partialuj} follows from \cref{{eqn:ujC}}. The remaining conditions are primal and dual feasibility (PF, DF) conditions and complementary slackness (CS) conditions from the Karush-Kuhn-Tucker (KKT) theorem as applied to linear programming problems. \hfill\qed
\end{proof}
\begin{corollary} A point $\mathbf{x}$ is an equilibrium if and only if there are vectors $\bm{\lambda},\bm{\mu} \in \mathbb{R}^n$ and the triple $(\mathbf{x},\bm{\lambda},\bm{\mu})$ is a global optimal solution to the following non-linear programming problem:
\begin{equation}
\left\{
\begin{aligned}
\min\;\; & \sum_j \lambda_j x_j + \mu_j(1-x_j) = \sum_j \mu_j-x_j\left(C^{(j)}_1(\mathbf{x}_{-j}) - C^{(j)}_0(\mathbf{x}_{-j})\right)\\
s.t.\;\; & C^{(j)}_1(\mathbf{x}_{-j}) - C^{(j)}_0(\mathbf{x}_{-j}) + \lambda_j - \mu_j = 0 \qquad \forall j\\
&\lambda_j \geq 0 \qquad \forall j\\
&\mu_j \geq 0 \qquad \forall j\\
&0 \leq x_j \leq 1
\end{aligned}
\right..
\end{equation}
Furthermore every global optimal solution has objective function value exactly equal to $0$.
\end{corollary}
\begin{proof} The proof is a specialization of the argument given in Chapter 6 of \cite{MS16}. In particular, note that the feasible conditions enforce the inequalities:
\begin{align*}
\lambda_j x_j \geq 0\\
\mu_j(1-x_j) \geq 0
\end{align*}
Therefore, the objective function is bounded below by $0$. If:
\begin{displaymath}
\sum_j \lambda_j x_j + \mu_j(1-x_j) = 0,
\end{displaymath}
then $\lambda_j x_j = 0$ and $\mu_j(1-x_j) = 0$ for all $j$. When taken with the other constraints, this implies that the triple $(\mathbf{x},\bm{\lambda},\bm{\mu})$ is a KKT point as given in \cref{prop:NeccessarySufficient}. Finally, we see that:
\begin{displaymath}
C^{(j)}_1(\mathbf{x}_{-j}) - C^{(j)}_0(\mathbf{x}_{-j}) = \mu_j - \lambda_j. 
\end{displaymath}
Algebraic manipulation shows that:
\begin{equation}
\sum_j \lambda_j x_j + \mu_j(1-x_j) = \sum_j \mu_j-x_j\left(C^{(j)}_1(\mathbf{x}_{-j}) - C^{(j)}_0(\mathbf{x}_{-j})\right)
\end{equation}
\hfill\qed
\end{proof}
We note that the KKT conditions of \cref{prop:NeccessarySufficient} can also be transformed into a complementarity problem \cite{CM96} and solved accordingly. Phrasing the problem as a non-linear programming problem allows for solution of small-scale examples using readily available software packages. 

We show that pure strategy equilibria exist for this game. The following sufficient condition ensures there is at least one pure strategy equilibrium.
\begin{proposition} Assume $\beta_1 \leq \beta_2 \leq \cdots \leq \beta_n$ and that:
\begin{align*}
\beta_m &\leq Am^\gamma - A(m-1)^\gamma\\
\beta_{m+1} & \geq A(m+1)^\gamma - Am^\gamma
\end{align*}
Then the point $x_1 = x_2 = \cdots x_m = 1$ and $x_{m+1} = x_{i+2} = \cdots = x_{n} = 0$ is an equilibrium in pure strategies. 
\end{proposition}
\begin{proof} The payoff to User $j$ is:
\begin{displaymath}
u_j = 
\begin{cases}
Am^\gamma - \beta_j & \text{if $1 \leq j \leq m$}\\
Am^\gamma & \text{otherwise}
\end{cases}.
\end{displaymath}
Suppose $1 \leq j \leq m$ and User $j$ unilaterally alters her strategy to $x_j < 1$. Then her new expected payoff is:
\begin{displaymath}
u_j' = Am^\gamma x_j - \beta x_j + (1-x_j)A(m-1)^\gamma.
\end{displaymath}
Compute:
\begin{multline*}
\Delta u_j = u_j' - u_j = (1-x_j)A(m-1)^\gamma-(1-x_j)\left(Am^\gamma - \beta_j\right) = \\
(1-x_j)\left(A(m-1)^\gamma - Am^\gamma + \beta_j\right) \leq 0,
\end{multline*}
by assumption. Thus User $j$ derives no benefit by unilaterally changing her strategy. Now assume $j \geq m+1$. Then:
\begin{displaymath}
u_j' = A(m+1)^\gamma x_j - \beta_j x_j + (1-x_j)Am^\gamma.
\end{displaymath}
Compute:
\begin{displaymath}
\Delta u_j = A(m+1)^\gamma x_j - \beta_j x_j - Am^\gamma x_j = \\
x_j\left(A(m+1)^\gamma - Am^\gamma - \beta_j\right) \leq 0,
\end{displaymath}
by assumption. Thus User $j$ derives no benefit by unilaterally changing her strategy. Therefore, the point $x_1 = x_2 = \cdots x_m = 1$ and $x_{m+1} = x_{i+2} = \cdots = x_{n} = 0$ is an equilibrium in pure strategies. 
\hfill\qed
\end{proof}

Because there may be many solutions to the KKT conditions from \cref{prop:NeccessarySufficient}, there may be mixed strategies even if the sufficient conditions are met. However, we can construct both necessary and sufficient conditions for pure strategy equilibria in which all users either share personal information or withhold personal information.

\begin{proposition} The strategy  $\mathbf{x} = \mathbf{0}$ is an equilibrium if and only if $A \leq \beta_1$.
\end{proposition}
\begin{proof} If $\mathbf{x} = \mathbf{0}$ is an equilibrium, then $\mu_j = 0$ for all $j$ and:
\begin{displaymath}
C^{(j)}_1(\mathbf{x}_{-j}) - C^{(j)}_0(\mathbf{x}_{-j}) = -\lambda \qquad \forall j.
\end{displaymath}
Correcting for the fact that $\mathbf{x}$ is on the boundary we see:
\begin{align*}
C^{(j)}_1(\mathbf{x}_{-j}) &= A - \beta_j\\
C^{(j)}_0(\mathbf{x}_{-j}) &= 0
\end{align*}
Thus, $A - \beta_j = -\lambda \leq 0$. It follows \textit{a fortiori} that $A \leq \beta_1$. 

Now suppose that $A < \beta_1$ and consider the strategy $\mathbf{x} = \mathbf{0}$. All users receive payoff $0$. Suppose User $j$ unilaterally changes her strategy to $x_j > 0$. Then her expected payoff is:
\begin{displaymath}
Ax_j - \beta_jx_j = (A - \beta_j)x_j \leq 0,
\end{displaymath}
because $A \leq \beta_1$ implying $A \leq \beta_j$ for all $j$.
Consequently no player has any incentive to unilaterally change strategy and $\mathbf{x}$ is an equilibrium. 
\hfill\qed
\end{proof}
By a similar argument, we have:
\begin{proposition} The strategy $\mathbf{x} = \mathbf{1}$ is an equilibrium if and only if $A \geq \beta_n$. \hfill\qed
\end{proposition}
These results yield a sensible interpretation for the parameter $A$. If $\beta_j$ is the perceived social cost of sharing personal information, then $A$ is a common perceived social benefit of sharing information and the decision to share or not becomes a simple cost-benefit analysis on the part of the user. 

In practice, it is rare that all users in a thread will share personal information. Moreover, users may not consistently share (or withhold) personal information, as illustrated in \cref{sec:Experiment}. Consequently, mixed strategies may be common (as illustrated in \cref{sec:Experiment}) or $A$ and $\beta_j$ ($j=1,\dots, n$) may be context-dependent.

\section{Experimental Results}\label{sec:Experiment}

Using the data set described in \cref{sec:Data}, we test our hypothesis that the number of comments (i.e., common reward) in a news posting game is modeled by \cref{eqn:CommonReward}. Articles with no comments were removed as they yield no additional information. This left 1977 articles for analysis. The proposed model is statistically significant above the $7\sigma$ level. \cref{tab:Params} provides confidence information on the parameters of the model.
\begin{table}[h!]
\centering
\begin{tabular}{|l|c|c|c|}
\hline
\textbf{Parameter} & \textbf{Value} & \textbf{$p$-Value} & \textbf{Confid. Ival.}\\
\hline
$\log(A)$ & $2.20$ & $0$ & $(2.15,2.25)$\\
\hline
$\gamma$ & $0.71$ & $1.4\times 10^{-312}$ & $(0.68,0.74)$\\
\hline
\end{tabular}
\caption{Parameters of the problem and confidence values}
\label{tab:Params}
\end{table}
The model explains $51\%$ of the variance in the observed data (i.e., $r^2-\mathrm{Adj} \approx 0.51$). \cref{fig:FitPlot} illustrates the fit of the data to the proposed model. 
\begin{figure}[htbp]
\centering
\includegraphics[width=0.7\textwidth]{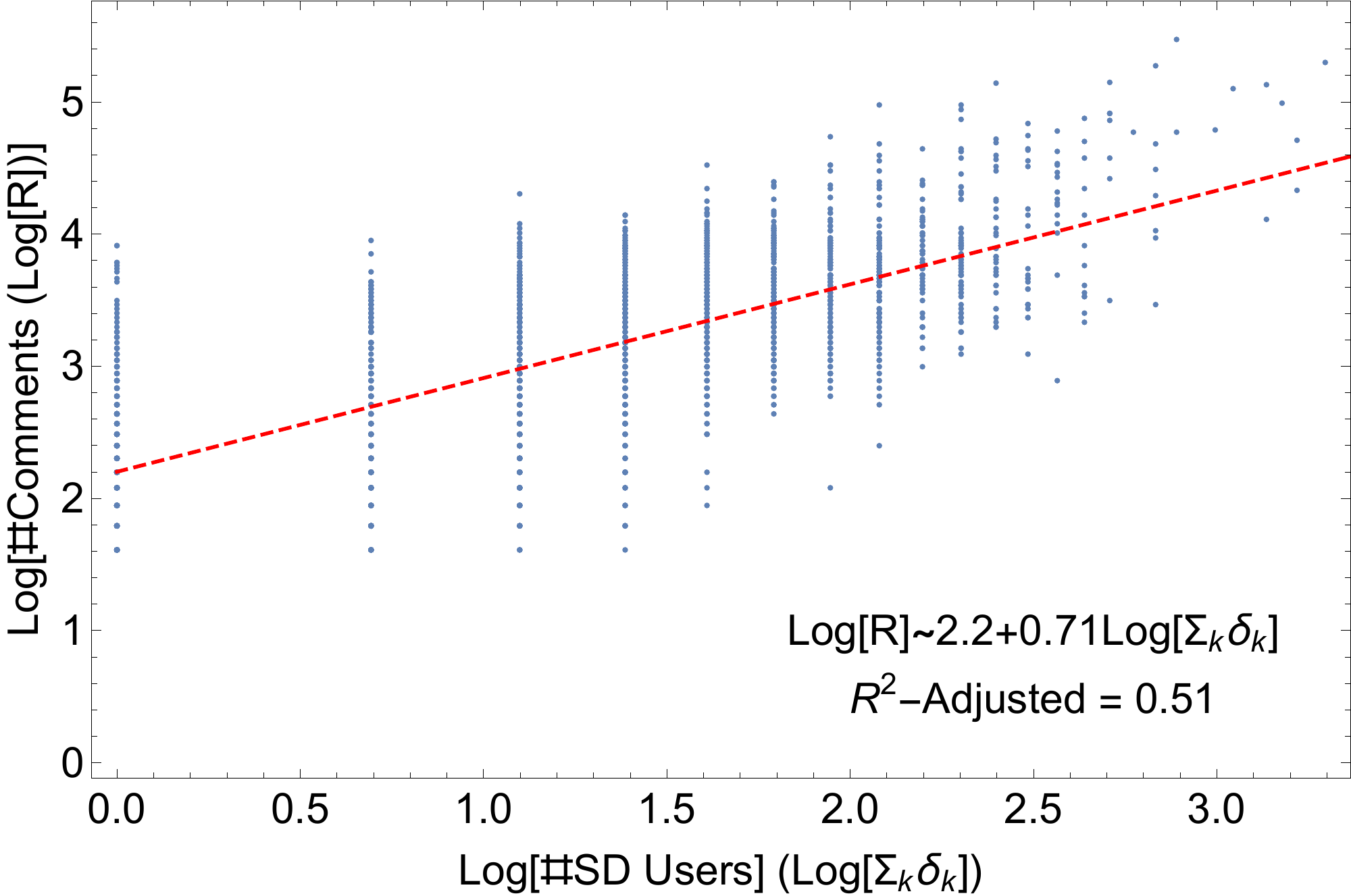}
\caption{We illustrate the goodness of the fit for the power law model, \cref{eqn:CommonReward}. A log-log scale is used.}
\label{fig:FitPlot}
\end{figure} 
The residual distribution is centered about zero and the $Q-Q$ plot illustrates reasonable normality of the residual distribution (see \cref{fig:Residuals}).
\begin{figure}[htbp]
\centering
\subfloat[Residual Histogram]{\includegraphics[width=0.45\textwidth]{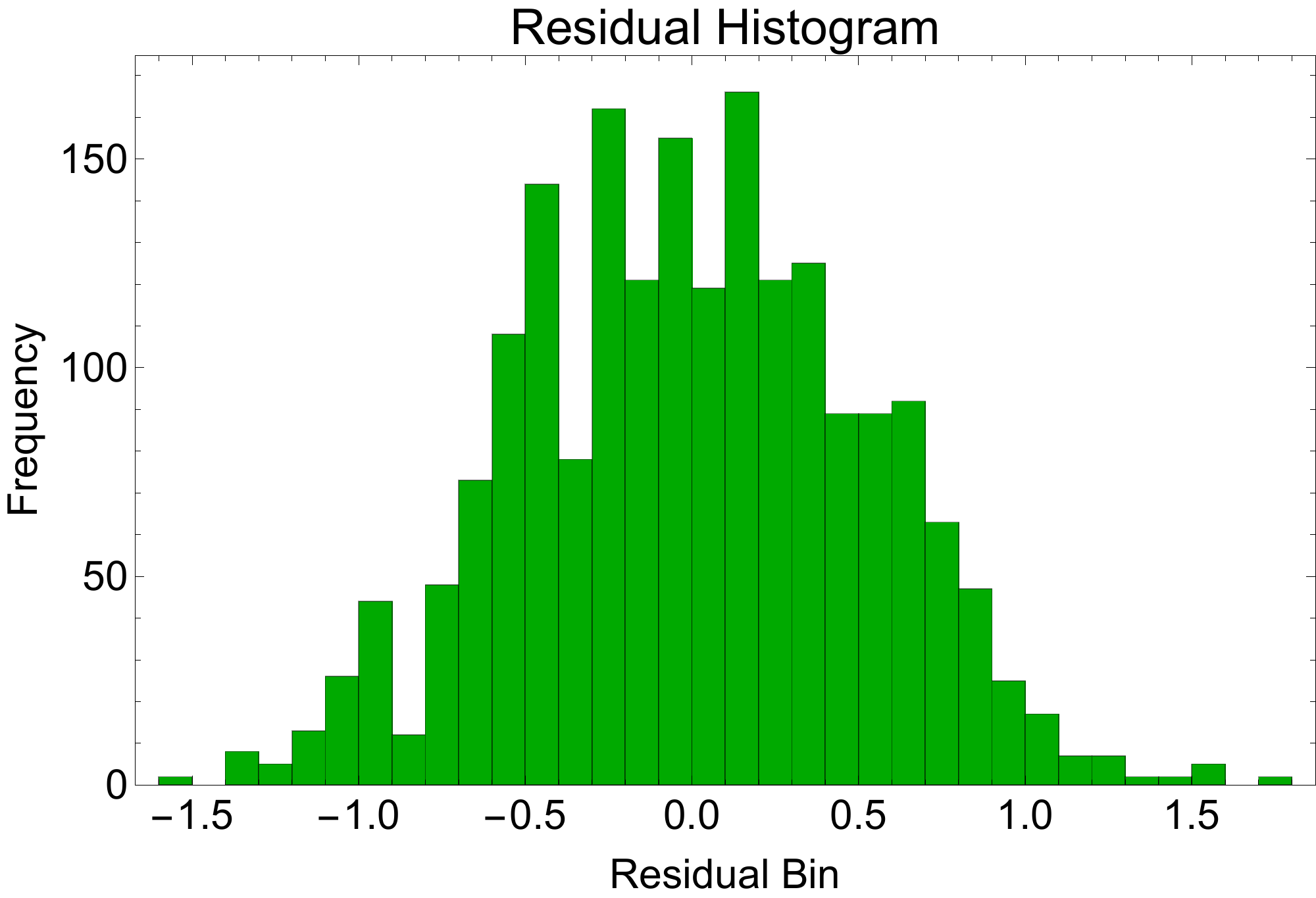}} \qquad
\subfloat[Residual $Q-Q$ plot]{\includegraphics[width=0.45\textwidth]{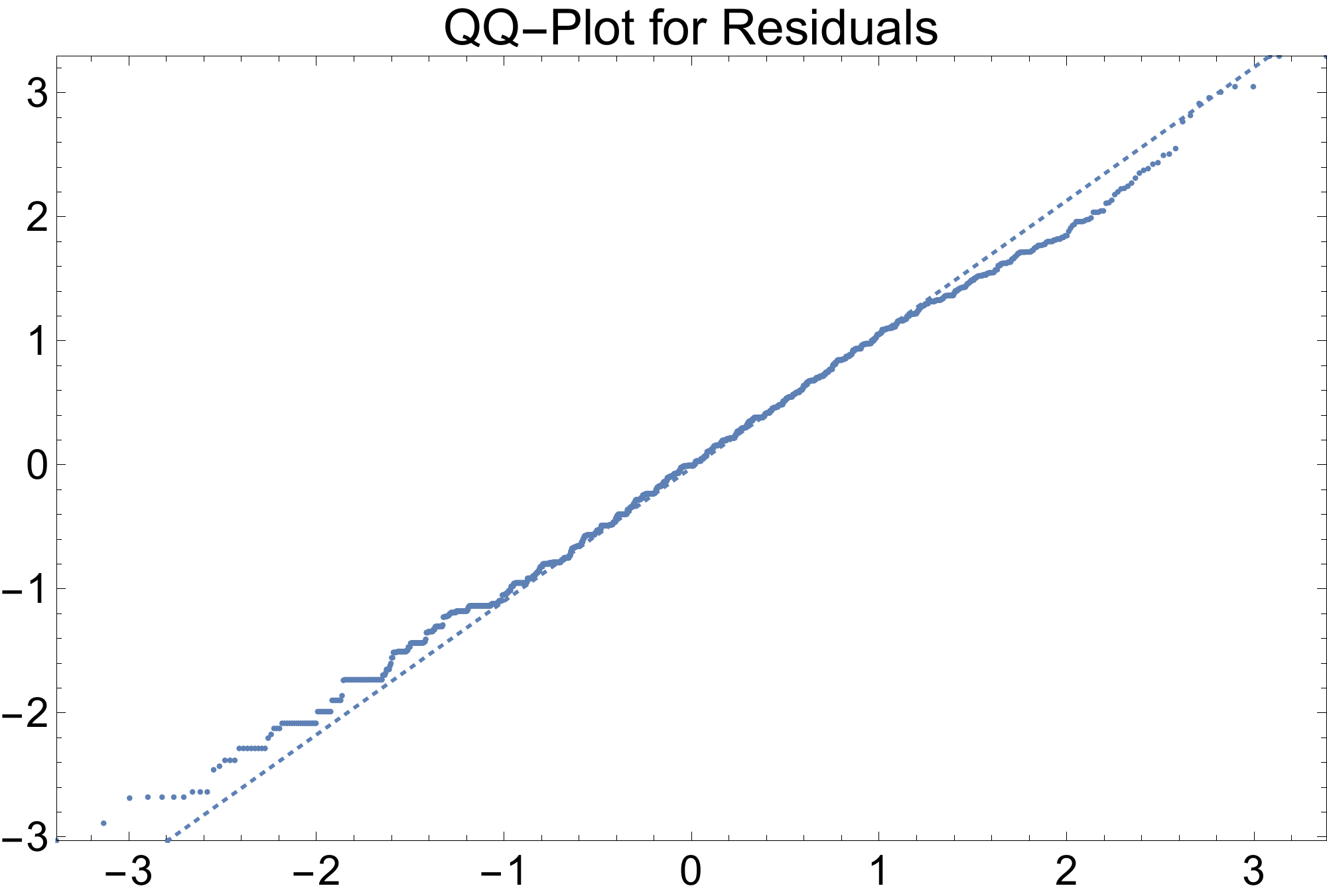}} 
\caption{(a) The histogram of the fit residuals using \cref{eqn:CommonReward} illustrates symmetry about $0$. (b) The $Q-Q$ plot illustrates approximate normality of the residuals.}
\label{fig:Residuals}
\end{figure}
Normality tests of the residuals showed mixed results with five of eight tests performed rejecting normality and the remaining tests failing to reject normality. Raw output of distribution testing is given below:
\begin{displaymath}
\begin{array}{l|ll}
 \text{} & \text{{\bf Statistic}} & \text{{\bf $p$-Value}} \\
\hline
 \text{Anderson-Darling} & 1.29886 & 0.00198841 \\
 \text{Baringhaus-Henze} & 1.88282 & 0.00824273 \\
 \text{Cram{\' e}r-von Mises} & 0.18939 & 0.0067714 \\
 \text{Jarque-Bera ALM} & 3.51543 & 0.168992 \\
 \text{Mardia Combined} & 3.51543 & 0.168992 \\
 \text{Mardia Kurtosis} & -1.88623 & 0.0592636 \\
 \text{Mardia Skewness} & 0.0342508 & 0.853174 \\
 \text{Pearson }\chi ^2 & 135.343 & \text{$1.465\times 10^{-12}$} \\
 \text{Shapiro-Wilk} & 0.997751 & 0.00679984. \\
\end{array}
\end{displaymath}

\subsection{Null-Model Comparison}
We compare the fit of the proposed model against the fit of a null linear model:
\begin{displaymath}
R_i \sim \beta_0 + \beta_1\sum_k(\delta_{i_k}).
\end{displaymath}
This model is also statistically significant above $7\sigma$ and explains $53\%$ of the variance (i.e., $r^2-\mathrm{Adj} = 0.53$). The fit, illustrating correlation is shown in \cref{fig:NullModel}. However, the residual distribution is decidedly not normal as illustrated by the $Q-Q$ plot (\cref{fig:NullResid}). This suggests that linear correlation is not the best explanation for the observed phenomena and supports our underlying hypothesis. 
\begin{figure}[htbp]
\centering
\subfloat[Null Model Fit]{\label{fig:NullModel}\includegraphics[width=0.45\textwidth]{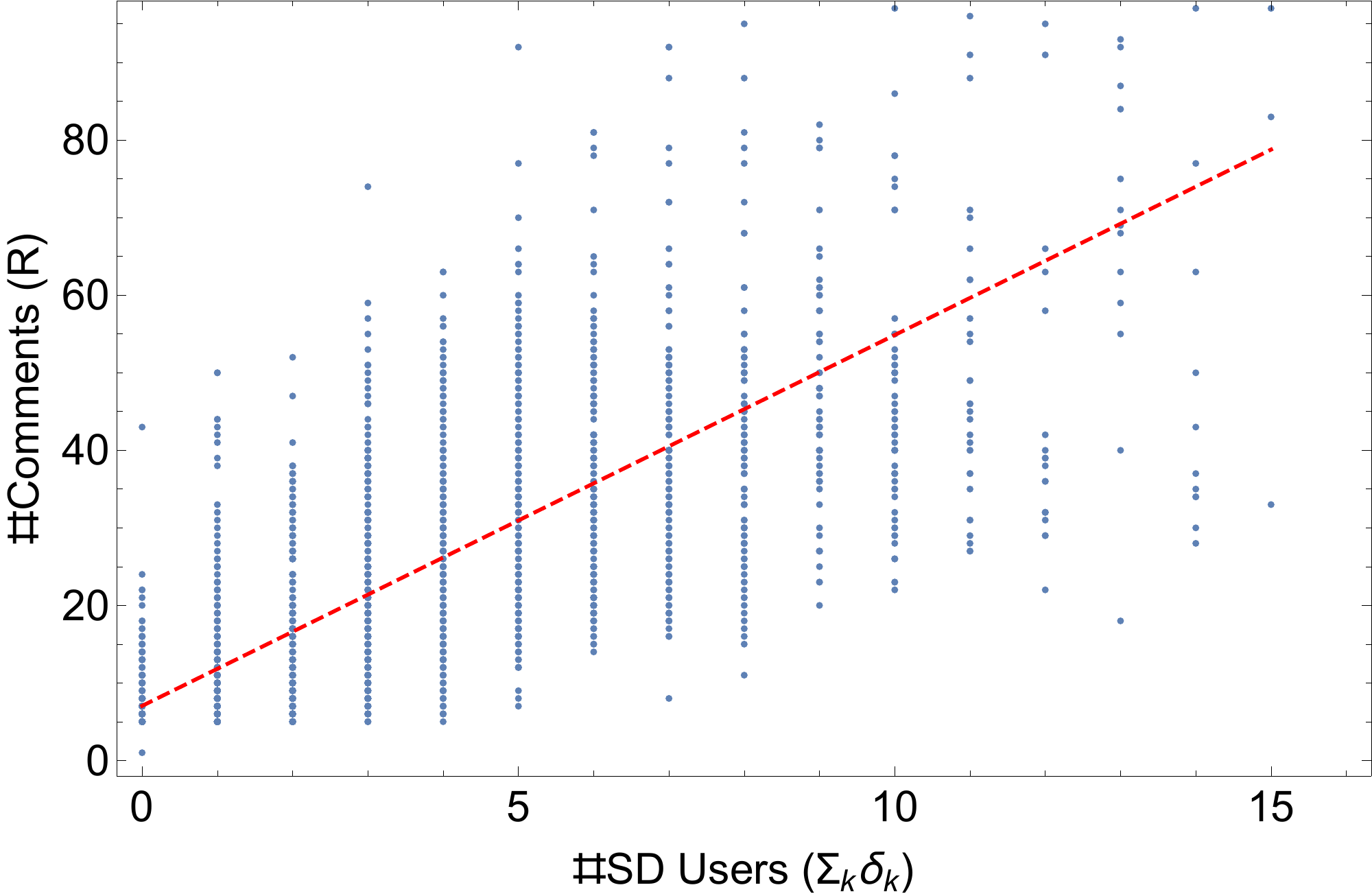}} \qquad
\subfloat[Null Model $Q-Q$ Plot]{\label{fig:NullResid}\includegraphics[width=0.45\textwidth]{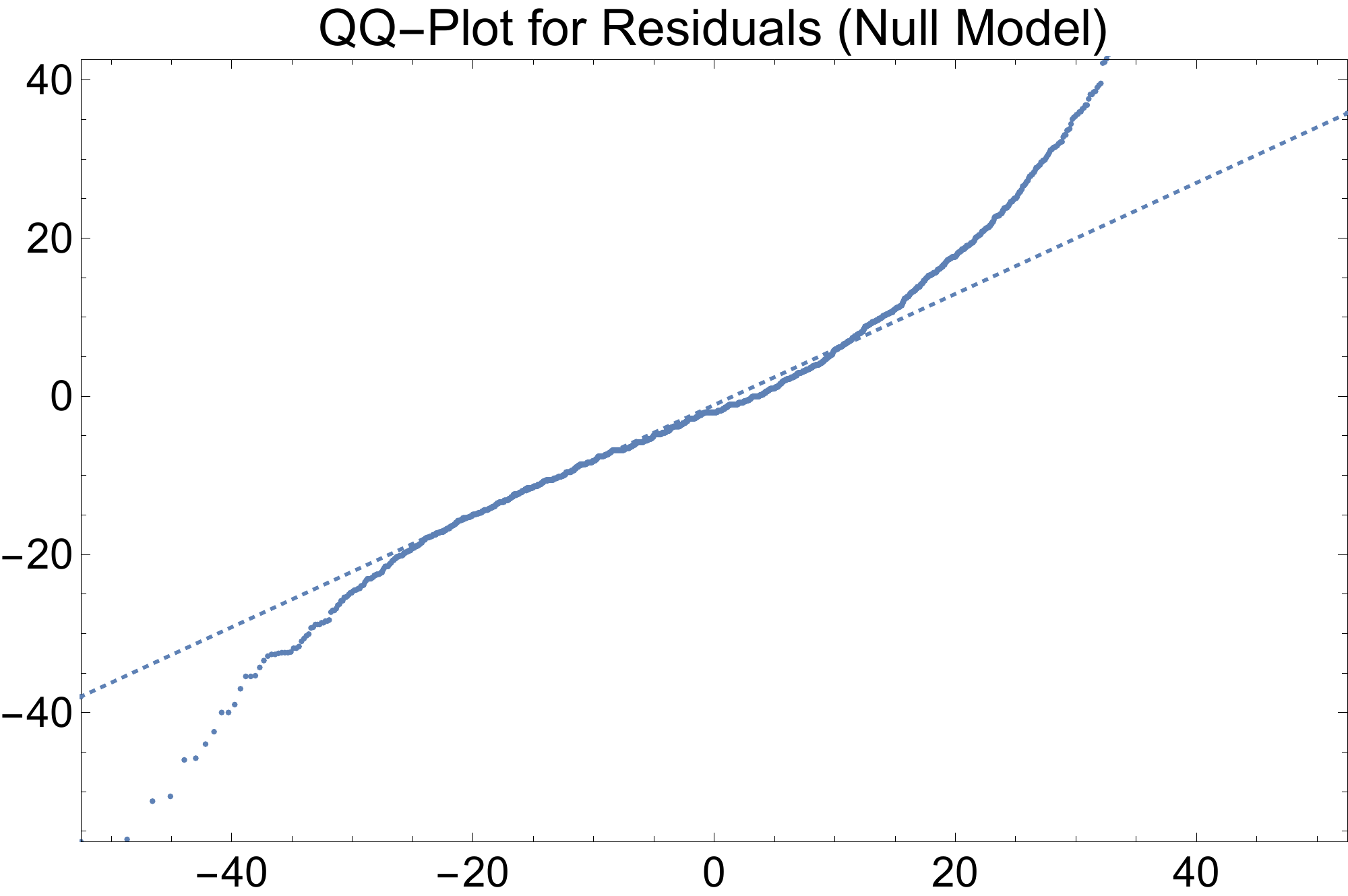}}
\caption{(a) We illustrate goodness of fit of the null (linear) model. (b) The $Q-Q$ plot of residuals of the null model illustrates clear lack of normality.}
\label{fig:NullModelAll}
\end{figure}

In addition to comparing relative fits, we also note that the AIC for the null model is $18,289.4$, while the AIC for the proposed model is $3025.78$, suggesting much better model parsimony for the proposed model over the null model.

\subsection{Fitting $\beta_j$: A Pilot Study}
As noted, this data set is not longitudinal and only a small number of users are repeat posters. This makes it impossible to estimate either $x_j$ or $\beta_j$ for all users. However, there are a subset of users who are repeat posters making it possible to estimate their mixed strategies and consequently their $\beta_j$. We outline the algorithm for this process and discuss results. This algorithm works particularly well when all players are using a mixed strategy. We note results in the remainder of this section are \textit{preliminary} and this should be considered as pilot data.
\begin{enumerate}
\item Compute $x_j$ using standard the standard MLE proportion estimator:
\begin{equation}
\hat{x}_j = \frac{\text{Number of Self Disclosing Posts}}{\text{Number of Posts}}.
\label{eqn:xhat}
\end{equation}
\item From \cref{eqn:DF} at equilibrium we must have:
\begin{displaymath}
C^{(j)}_1(\hat{\mathbf{x}}_{-j}) - C^{(j)}_0(\hat{\mathbf{x}}_{-j}) = \mu_j - \lambda_j, \qquad \forall j.
\end{displaymath}
These equations can be used to fit an estimate for $\beta_j + \mu_j - \lambda_j$. 
\end{enumerate}
In particular, when $\hat{x}_j \in (0,1)$, then $\lambda_j = \mu_j = 0$ and:
\begin{equation}
\hat{\beta_j} = \sum_{\bm{\delta}_{-j}} A\left(\left(1 + \sum_{k \neq j}\delta_k\right)^\gamma - \left(\sum_{k \neq j}\delta_k\right)^\gamma\right)\prod_{k\neq j} \hat{x}_k^{\delta_k}(1-\hat{x}_k)^{(1-\delta_k)}.
\label{eqn:betajstats}
\end{equation}
If there are several articles (each with different number of users, $N$), then $\hat{\beta}_j$ is computed over all instances of \cref{eqn:betajstats} and the mean is the MLE of $\hat{\beta}_j$. In our analysis, $\hat{x}_j$ was \textit{not} available for all users (because of data limitations). In analyzing an article with users who did not have a proper $\hat{x}_j$, the mean of all available $\hat{x}_j$ was substituted. We denote this mean $\bar{x}$. 

In cases with articles where several user strategies were estimated with $\bar{x}$, we restricted the analysis to size $N = 8$ users for computational speed. By the central limit theorem, this approximation will not affect the resulting estimates of $\hat{b}_j$ substantially. Put more simply: An article with 34 users requires computing a sum with $2^{33}$ summands. If $30$ users are estimated as $\bar{x}$, we used only $8$ of those users in computing \cref{eqn:betajstats}. 

Using this approach we estimated the strategy for all users in the data set who posted to at least 15 articles. We used parameter estimates for $A$ and $\gamma$ obtained in the previous section. There were 135 users in this subsample. A histogram of their estimated strategies is shown in \cref{fig:hatX}.
\begin{figure}[htbp]
\centering
\subfloat[Histogram of $\hat{x}_j$]{\label{fig:hatX}\includegraphics[width=0.45\textwidth]{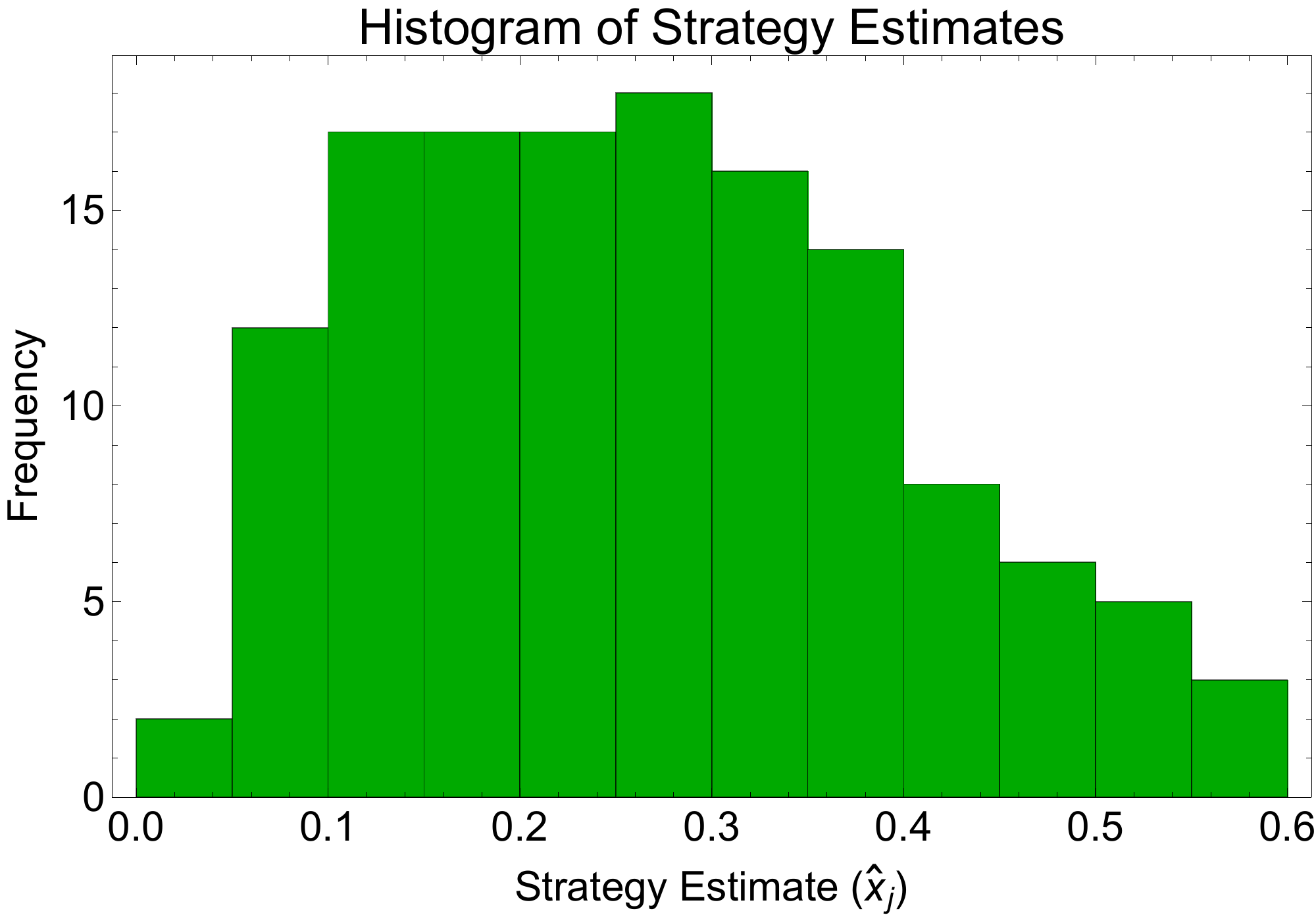}} \qquad
\subfloat[Histogram of $\hat{\beta}_j$]{\label{fig:hatBeta}\includegraphics[width=0.45\textwidth]{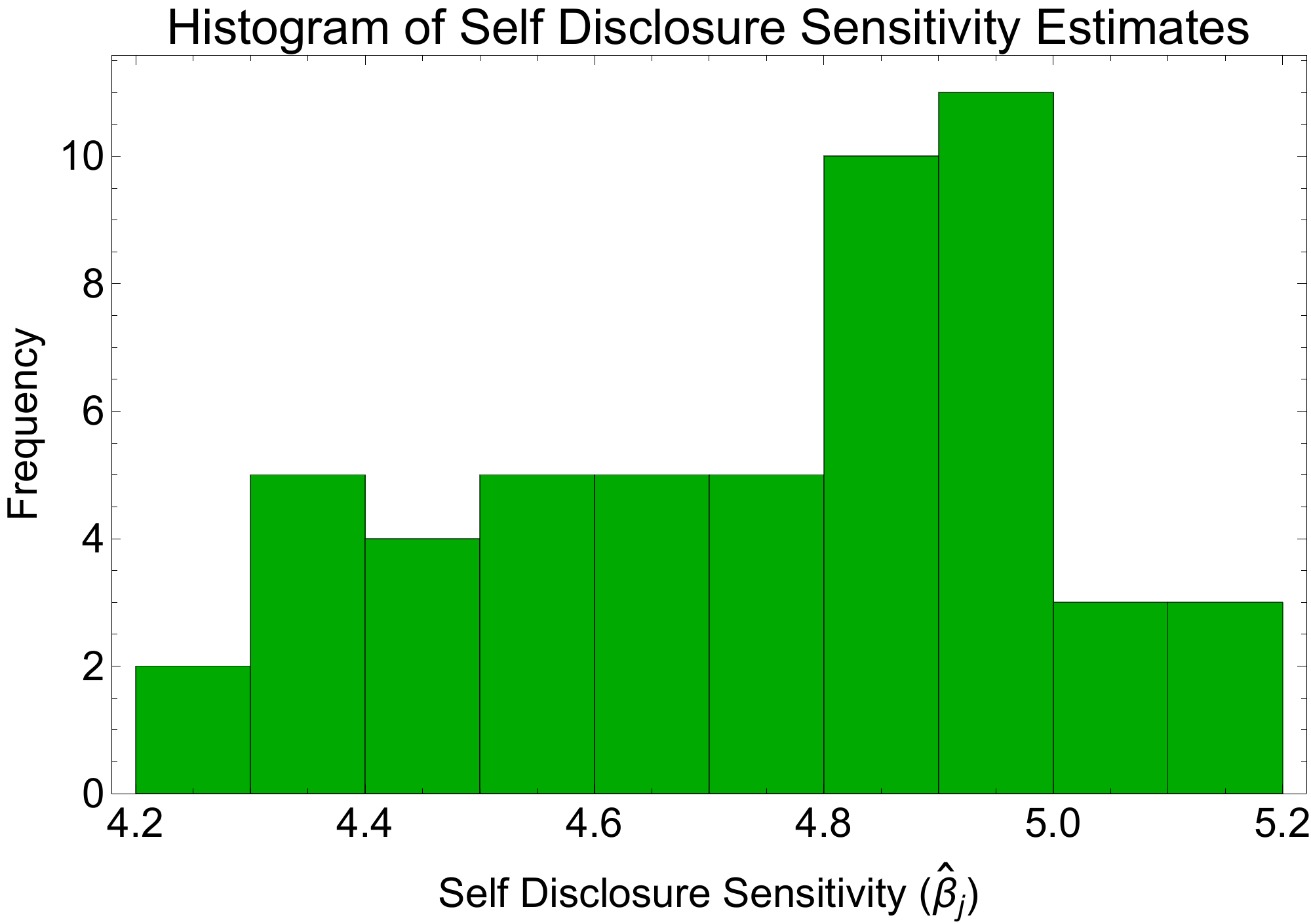}}
\caption{(a) The histogram of the strategies suggests an almost uniform distribution between $\hat{x} = 0.1$ and $\hat{x} = 0.4$ with a sharp dropoff after that. (b) Similarly, $\beta_j$ shows an almost uniform distribution with a high concentration of values near $4.8$. This histogram is drawn from a small sample size, so we exercise care in interpreting these results.}
\label{fig:XBetaHistograms}
\end{figure}
In particular, all estimated strategies were mixed, suggesting that pure strategies, while possible, are less likely to occur in real data. Using \cref{eqn:betajstats}, we estimated $\hat{\beta}$ in articles containing at least three users for whom $\hat{x}_j$ had been estimated. We estimated $\hat{\beta}_j$ for 14 users who had posted at least 15 times and who had posted in at least one article with 2 other such users. The histogram of these estimates is given in \cref{fig:hatBeta}. 

To validate the hypothesis that higher $\beta_j$ is correlated with lower $x_j$, we performed a simple linear fit, which is shown in \cref{fig:XvsBeta}
\begin{figure}[htbp]
\centering
\includegraphics[width=0.65\textwidth]{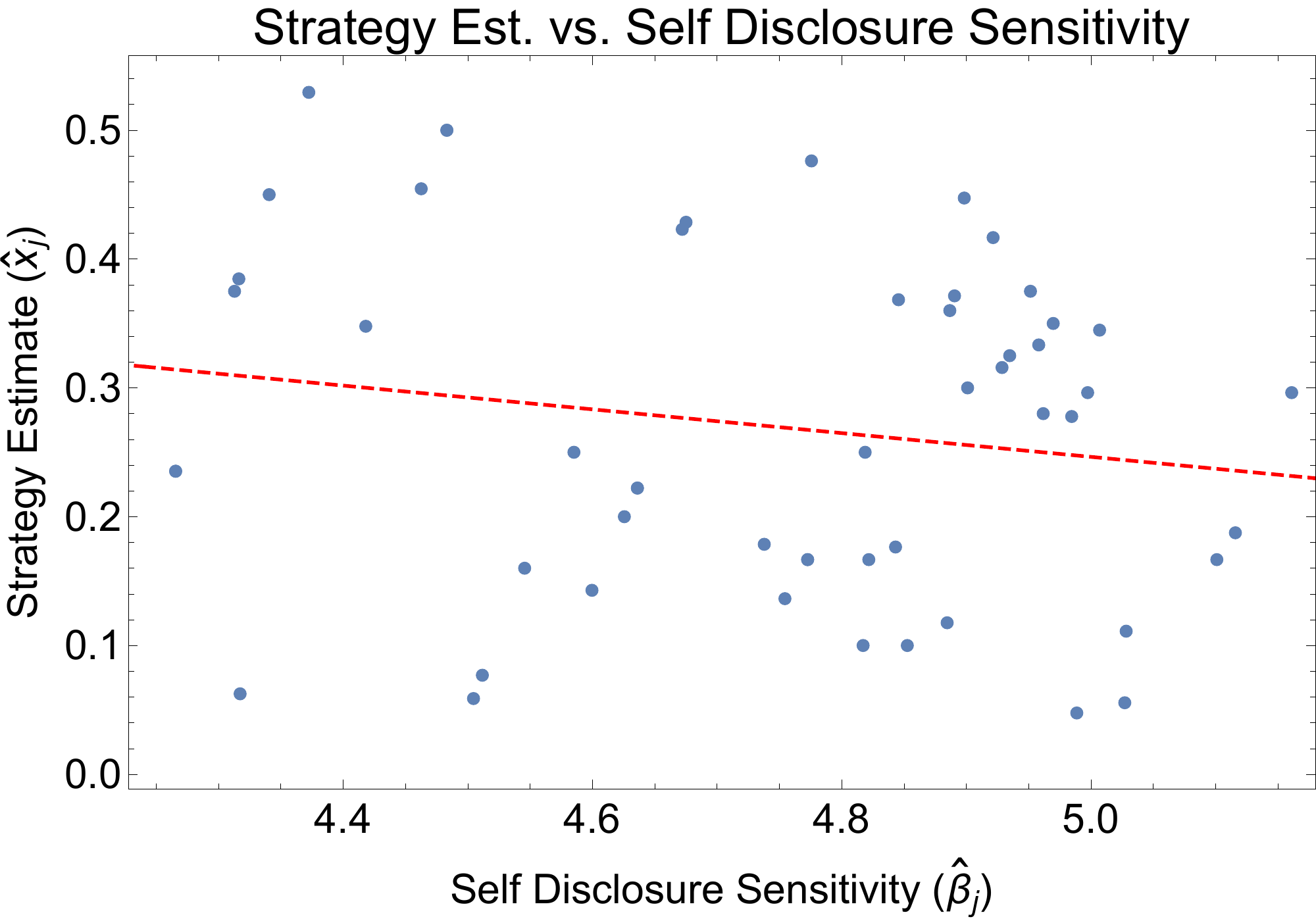}
\caption{We illustrate the correlation between $\hat{x}_j$ and $\hat{\beta}_j$. We expect $\hat{x}_j \sim a_0 - a_1\beta_j$ with $a_1 > 0$, which we see.}
\label{fig:XvsBeta}
\end{figure}
A table of parameter values and confidence regions are shown below, with the dependent variable being $\hat{x}_j$ and the independent variable $\hat{\beta}_j$.
\begin{displaymath}
\begin{array}{l|llll}
 \text{} & \text{Estimate} & \text{Standard Error} & \text{t-Statistic} & \text{P-Value} \\
\hline
 1 & 0.707635 & 0.354241 & 1.99761 & 0.0511066 \\
 \hat{\beta} & -0.0922398 & 0.0745993 & -1.23647 & 0.221948 \\
\end{array}
\end{displaymath}
The coefficient of $\hat{\beta}$ is negative (as predicted). However, the model is only 
significant at $p\approx 0.22$. This is far too high to be considered conclusive, but is suggestive that additional data collection and analysis may be warranted.

\subsection{Alternate Method for Fitting $\hat{\beta}$}
Inspired by techniques from crystallography and neutron scattering techniques \cite{IJ09}, we also propose an alternate fitting method, the analysis of which will be considered future work. Given an estimator $\hat{\mathbf{y}}$ for individual user strategies, we define a complementarity constrained least-squares fit problem: 
\begin{equation}
\left\{
\begin{aligned}
\min_{\mathbf{x},\bm{\lambda}, \bm{\mu},\bm{\beta}} \;\; & \sum_j (\hat{y}_j - x_j)^2\\
& s.t. \;\; x_j \leq 1 \qquad \forall j\\
& x_j \geq 0 \qquad \forall j\\
&\sum_{\bm{\delta}_{-j}} A\left(\left(1 + \sum_{k \neq j}\delta_k\right)^\gamma - \left(\sum_{k \neq j}\delta_k\right)^\gamma\right)\prod_{k\neq j} \hat{x}_k^{\delta_k}(1-\hat{x}_k)^{(1-\delta_k)} + \\
&\hspace*{18em} \lambda_j - \mu_j - \beta_j= 0 \qquad \forall j\\
& \lambda_j \geq 0 \qquad \forall j\\
& \mu_j \geq 0 \qquad \forall j\\
&\lambda_j x_j = 0 \qquad \forall j\\
&\mu_j (x_j - 1)  = 0 \qquad \forall j.\end{aligned}
\right.
\end{equation}
Solutions to this least square fit are tuples $(\mathbf{x},\bm{\lambda}, \bm{\mu},\bm{\beta})$ where $\beta_j$ ($j=1,\dots,N)$ are now unknowns. By \cref{prop:NeccessarySufficient}, any vector $\mathbf{x}$ in such a solution is a Nash equilibrium for the derived vector $\bm{\beta}$. Moreover, the derived equilibrium minimizes the least square error with respect to the estimated equilibrium $\hat{\mathbf{y}}$, derived from \cref{eqn:xhat}. As a mixed complementarity problem, the proposed constrained least squares estimator is challenging to solve \cite{DF95,BDF97}, since these problems are known to be NP-hard in general. Considering the limitations of the data, we reserve further analysis of this approach for future work with a more complete data set. 

\section{Conclusions and Future Directions}\label{sec:Conclusion}
In this paper, we have proposed a public goods model of personal information disclosure in news article commentaries. We have found sufficient conditions for the proposed public goods game to exhibit pure strategy equilibria and showed that for any choice of model parameters, there is always at least one equilibrium. Special necessary and sufficient conditions were identified for the case in which all users choose not to disclose personal information or when all users choose to disclose (some) personal information. We have validated this model using a dataset of online comments on news outlets and showed that the proposed common reward function fits the underlying data set better than a null (linear) model. For a small subset of users, we have estimated their strategy ($\hat{x}_j)$ as well as their sensitivity to personal information disclosure ($\hat{\beta}_j$). We consider this a pilot study because the publicly available data set used in this study was not longitudinal, thus limiting our ability to study a large population of users over time.

In future work, we will determine whether this model is valid using larger data sets when available. In particular, we have proposed a fitting approach for determining 
$\hat{\beta}_j$ that relies on the solution to a large-scale mixed complementarity problem. Studying this fitting problem, its complexity, and results from its application form the foundation of future work. In addition to this, we may investigate other commenting environments in which users may choose to share personal information to further validate this model and determine whether it holds across a broad spectrum of online platforms.

\section*{Acknowledgements}
Portions of Griffin's work were supported by the National Science Foundation under grant CMMI-1463482. Griffin also wishes to thank A. Belmonte for the helpful discussion. 
Dr. Squicciarini's work is partially funded by the National Science Foundation under grant 1453080. Dr. Squicciarini and Dr. Rajtmajer are also partially supported by PSU Seed grant 425-02. 

\bibliographystyle{IEEEtran}
\bibliography{PubGoodsModelBib}

\end{document}